\documentclass[a4paper, 11pt]{article}

\usepackage[margin=3cm]{geometry}
\usepackage{hyperref,xcolor}
\usepackage{amsthm, amssymb, amsmath}
\usepackage{enumerate}
\usepackage{algorithm}
\usepackage{authblk}
\usepackage{float}
\usepackage{graphics}
\usepackage{graphicx}
\usepackage{pstricks}
\usepackage{subfigure}
\usepackage[noend]{algpseudocode}

\hypersetup{
    colorlinks=false,
	linkcolor=red,
	citecolor=red,
    pdfborder={0 0 0},
}

\title{Maximum weighted independent sets with a budget}

\author[1]{Tushar Kalra}
\author[1]{Rogers Mathew}
\author[1]{Sudebkumar Prasant Pal}
\author[1]{Vijay Pandey}
 \affil[1]{
 	Department of Computer Science and Engineering, \authorcr 
 	Indian Institute of Technology \authorcr
 	Kharagpur 721302, West Bengal, India. \authorcr
 	\texttt{\{tushar11nitjkalra, rogersmathew, sudebkumar, vijayiitkgp13\}@gmail.com}
 }

\theoremstyle{definition}
\newtheorem{definition}{Definition}

\theoremstyle{plain}
\newtheorem{theorem}{Theorem}
\newtheorem{lemma}[theorem]{Lemma}
\newtheorem{construction}[theorem]{Construction}
\newtheorem{corollary}[theorem]{Corollary}

\newtheorem{observation}[theorem]{Observation}

\newtheorem{example}[theorem]{Example}

\theoremstyle{remark}

\newtheoremstyle{plainitshape}
  {}
  {}
  {\itshape}
  {}
  {\itshape}
  {.}
  {0.5em}
  {}
\theoremstyle{plainitshape}

\newtheoremstyle{cases}
  {}
  {}
  {}
  {}
  {}
  {\newline}
  {0.5em}
  {{\itshape \thmname{#1}} \thmnumber{#2} ({\itshape\thmnote{#3}}).\medskip}
\theoremstyle{cases}
\newtheorem{case}{Case}

\newtheoremstyle{constructions}
  {}
  {}
  {}
  {}
  {}
  {}
  {0.5em}
  {{\itshape \thmname{#1}} \thmnumber{#2}\medskip}
\theoremstyle{constructions}

\algnewcommand\algorithmicinput{\textbf{INPUT:}}
\algnewcommand\INPUT{\item[\algorithmicinput]}

\algnewcommand\algorithmicoutput{\textbf{OUTPUT:}}
\algnewcommand\OUTPUT{\item[\algorithmicoutput]}

\begin{document}
\date{}
\maketitle
\begin{abstract}
Given a graph $G$, a non-negative integer $k$, and a weight function that maps each vertex in $G$ to a positive real number, the \emph{Maximum Weighted Budgeted Independent Set (MWBIS) problem} is about finding a maximum weighted independent set in $G$ of cardinality at most $k$. A special case of MWBIS, when the weight assigned to each vertex is equal to its degree in $G$, is called the \emph{Maximum Independent Vertex Coverage (MIVC)} problem. In other words, the MIVC problem is about finding an independent set of cardinality at most $k$ with maximum coverage. 

Since it is a generalization of the well-known Maximum Weighted Independent Set (MWIS) problem, MWBIS too does not have any constant factor polynomial time approximation algorithm assuming $P \neq NP$. In this paper, we study MWBIS in the context of bipartite graphs. We show that, unlike MWIS, the MIVC (and thereby the MWBIS) problem in bipartite graphs is NP-hard. 
Then, we show that the MWBIS problem admits a $\frac{1}{2}$-factor approximation algorithm in the class of bipartite graphs, which matches the integrality gap of a natural LP relaxation. 
 
\end{abstract}
\section{Introduction}
\subsection{Problem definition}
Let $G$ be a graph and let $w:V(G) \rightarrow \mathbb{R}^+$ be a function that assigns positive real numbers as weights to the vertices of $G$. Under this assignment of weights, for any set $S \subseteq V(G)$, we define the weight of $S$, denoted by $w(S)$, as the sum of the weights of the vertices in $S$. The famous \emph{Maximum Weighted Independent Set (MWIS) problem} is about finding an independent set of vertices in $G$ that has the highest weight amongst all the independent sets in $G$. In this paper, we  study a budgeted version of the well-studied MWIS problem, namely \emph{Maximum Weighted Budgeted Independent Set} (MWBIS) problem. 
\begin{definition}
\label{def:MWBIS}
Given a graph $G$, a weight function $w: V(G) \rightarrow \mathbb{R}^+$, and a positive integer $k$, the \emph{MWBIS problem} is about finding an independent set of size at most $k$ in $G$ that has the highest weight amongst all independent sets of size at most $k$ in $G$.
\end{definition}
\subsection{Related work}
A more general problem (also known by the same name, MWBIS,) was introduced and studied in the context of special graphs like trees, forests, cycle graphs, interval graphs, and planar graphs in \cite{Bandyapadhyay14} where each vertex in the given graph $G$ has a cost associated with it and the problem is about finding an independent set of total cost at most $C$ (, where $C$ is a part of the input,) in $G$ that has the highest weight amongst all such independent sets. Apart from this work, to the best of our knowledge, not much is known about MWBIS.

Given a graph $G$, we know that the \emph{Vertex Cover (VC) problem} is about finding the minimum number of vertices that cover all the edges of $G$. Several variants of the VC problem has been studied in the literature. We discuss about a couple of them here. For a positive integer $t$, the \emph{Partial Vertex Cover (PVC) problem} is about finding the minimum number of vertices that cover at least $t$ distinct edges of $G$. In the year $1998$, Burroughs and Bshouty introduced and studied the problem of partial vertex cover \cite{bshouty1998massaging}. In this paper, the authors gave a $2$-factor approximation algorithm by rounding  fractional optimal solutions given by an LP relaxation of the problem. Bar-Yehuda in \cite{bar2001using} came up with another $2$-approximation algorithm that relied on the beautiful `local ratio' method. A primal-dual algorithm achieving the same approximation factor was given in \cite{gandhi2004approximation}. In \cite{caskurlu2014partial}, it was shown that the PVC problem on bipartite graphs is NP-hard. 

Another popular variant of the VC problem is the \emph{Maximum Vertex Coverage} (MVC) problem. Given a graph $G$ and a positive integer $k$, the \emph{MVC problem} is about finding $k$ vertices that maximize the number of distinct edges covered by them in $G$. Ageev and Sviridenko in \cite{ageev1999approximation} gave a $3/4$-approximation algorithm for the MVC problem. An approximation algorithm, that uses a semidefinite programming technique, based on a parameter whose factor of approximation is better than $3/4$ when the parameter is sufficiently large was shown in \cite{han2002approximation}. Apollonio and Simeone in \cite{apollonio2014maximum} proved that the MVC problem on bipartite graphs is NP-hard. The same authors in \cite{apollonio2014improved} gave a $4/5$ factor approximation algorithm for MVC on bipartite graphs that exploited the structure of the fractional optimal solutions of a linear programming formulation for the problem. The authors of \cite{caskurlu2014partial} improved this result to obtain an $8/9$ factor approximation algorithm for MVC on bipartite graphs. 

\subsection{MIVC problem - a special case of MWBIS}
Let us come back to our problem - the MWBIS problem. In this problem, what happens if the weight function given maps each vertex to its degree in $G$? That is, let $w(v) = deg(v),~\forall v \in V(G)$. 
We call this the \emph{Maximum Independent Vertex Coverage (MIVC) problem}. 
\begin{definition}
\label{def:MIVC1}
Given a graph $G$, a weight function $w: V(G) \rightarrow \mathbb{R}^+$ defined as, for each vertex $v \in V(G)$, $w(v) = deg(v)$, and a positive integer $k$, the MIVC problem is about finding an independent set of size $k$ in $G$ that has the highest weight amongst all independent sets of size at most $k$.
\end{definition}  
Observe that the MIVC problem, as its name suggests, can also be seen as a variant of the MVC problem where the $k$ vertices that we choose need to be an independent set. This observation gives us the following alternate definition. 
\begin{definition}
\label{def:MIVC2}
Given a graph $G$ and a positive integer $k$, the \emph{MIVC problem} is about finding at most $k$ independent vertices that maximize the number of edges covered by them in $G$. 
\end{definition}
\subsection{IP formulation of MWBIS}
\label{subsec:IPforMWBIS}
Let $G$ be a graph on $n$ vertices, where $V(G) = \{v_1, \ldots , v_n\}$, and let $w:V(G) \rightarrow \mathbb{R}^+$ be a weight function given. Let $k$ be a positive integer given. Let $\mathcal{C}$ denote the set of all maximal cliques in $G$. In order to formulate MWBIS on $G$ as an integer program, let us assign a variable $x_i$ for each $v_i \in V(G)$, which is allowed $0/1$ values. The variable $x_i$ will be set to $1$ if and only if the vertex $v_i$ is picked in the independent set.   
\begin{align}
Maximize &\sum_{i \in [n]} w(v_i)\cdot x_{i} \label{IP:MIVC} \\
&s.t. \sum_{i=1}^{n} x_{i} \leq k  \nonumber \\
&\sum_{i:v_i \in C}x_i \leq 1, \quad \forall C \in \mathcal{C} \nonumber \\
&x_{i} \in \lbrace 0,1 \rbrace , \quad \forall i \in [n]. \nonumber 
\end{align}
The constraint $\sum_{i=1}^{n} x_{i} \leq k$ ensures that not more than $k$ vertices are picked. The constraint $\sum_{i:v_i \in C}x_i \leq 1$, for each maximal clique $C \in \mathcal{C}$, ensures that the set of selected vertices is an independent set. Replacing $w(v_i)$ with $deg(v_i)$ in the objective function of (\ref{IP:MIVC}) yields the IP formulation of MIVC for $G$. 

Later in Section \ref{subsec:bipartiteIntGap}, we consider the LP-relaxation of this integer program. We show an instance of the MIVC problem (a bipartite graph) to illustrate the fact that the integrality gap of this LP relaxation is not `good'. This rules out the possibilities of using this relaxation to obtain good  LP-based approximation algorithms for the MIVC problem (and thereby for the MWBIS problem) in bipartite graphs.

\subsection{Hardness results}
It is well-known that the MWIS problem cannot be approximated to a constant factor in polynomial time, unless P=NP. For every sufficiently large $\Delta$, there is no $\Omega(\frac{\log^2\Delta}{\Delta})$-factor polynomial time approximation algorithm for MWIS in a degree-$\Delta$ bounded graph, assuming the unique games conjecture and $P \neq NP$ \cite{austrin2009inapproximability}. As MWBIS is a generalization of MWIS, these results hold   true for the MWBIS problem too. But, what about the special case of MWBIS - the MIVC problem?   
Below, we show that it is hard to approximate MIVC within a certain factor. We prove this by giving an approximation factor preserving reduction from another problem, namely $3$-Maximum Independent Set ($3$-MIS) problem. The \emph{$3$-MIS problem} is about finding an independent set of maximum cardinality in a given $3$-regular graph. 
\begin{theorem}
\label{thm:approxHardness}
There is no  $\left(\frac{139}{140} + \epsilon \right)$-factor approximation algorithm for the MIVC problem, assuming $P \neq NP$, where $\epsilon > 0$. 
\end{theorem}
\begin{proof}
It was shown in (Statement $(v)$ of Theorem $1$ in) \cite{karpinski1998some} that there is no  $\left(\frac{139}{140} + \epsilon \right)$-factor approximation algorithm for the $3$-MIS problem, assuming $P \neq NP$, where $\epsilon > 0$. If there existed a polynomial time $f$-factor approximation algorithm $A(G,k)$ for the MIVC problem, where $f > 139/140$, then we could use it as an $f$-factor approximation algorithm for the $3$-MIS problem as follows.  Given a positive integer $k$ and a $3$-regular graph $G$ on $n$ vertices as input, we run $A(G,n)$. This algorithm will return an independent set of cardinality at least $f$ times the size of a maximum independent set in $G$.  
\end{proof}
We define decision versions of the MWBIS and the MIVC problems as:-
\begin{eqnarray*}
\mbox{D-MWBIS} & = & \{ <G,w,k,t>~|~G \mbox{ is an undirected graph, } w:V(G) \rightarrow \mathbb{R}^+ \mbox{ is a } \\                & & \mbox{weight function and } G \mbox{ contains an independent } \mbox{set of cardinality } \\
& & \mbox{at most } k \mbox{ whose } \mbox{weight is at least } t \}. \\
\mbox{D-MIVC} & = & \{ <G,k,t>~|~G \mbox{ is an undirected graph and } G \mbox{ contains a set of } \\  
& & \mbox{at most } k \mbox{ independent vertices } \mbox{ that cover at least } t \mbox{ distinct } \\
& & \mbox{edges}\}.  
\end{eqnarray*}
Since the MIS problem and the VC problem are known to be NP-hard, it is not surprising to see that the D-MIVC (and thereby the D-MWBIS) problem is NP-hard for general graphs. The following is an easy corollary to Theorem \ref{thm:approxHardness}
\begin{corollary}
D-MIVC problem is NP-hard. 
\end{corollary}
\subsection{Motivation}
The main motivation behind studying the MWBIS problem is the fact that it is a natural generalization of a very well-studied problem in the context of graph algorithms and approximation algorithms, the MWIS problem. 
MWIS problem finds application in wireless networks, scheduling, molecular biology, pattern recognition, coding theory, etc. MWBIS finds application in most scenarios where MWIS is used. Refer \cite{Bandyapadhyay14} to know about applications of MWBIS in `job scheduling in a computer' and in `selecting non-interfering set of transmitters'.

\subsection{Our contribution}
We know that the MWIS problem on bipartite graphs is polynomial time solvable. Intuitively, the MWBIS problem, being a budgeted variant of the MWIS, is also expected to follow the same behaviour. But, contrary to our intuition, we show in Theorem \ref{thm:bipNPhard} that the D-MIVC (and thereby the D-MWBIS) problem is NP-hard for bipartite graphs. This motivated us into looking at approximation algorithms. In Section \ref{subsec:halfApprox}, we give an $O(nk)$ time, $\frac{1}{2}$-factor greedy approximation algorithm for the MWBIS problem on bipartite graphs. We give a tight example to the algorithm. In Section \ref{subsec:bipartiteIntGap}, we consider the LP relaxation of the integer program for MWBIS given in Section \ref{subsec:IPforMWBIS}. We show that the integrality gap for this LP is upper bounded by $\frac{1}{2} + \epsilon$ for bipartite graphs, where $\epsilon$ is any number greater than $0$. In other words, no LP-based technique, that uses this natural LP relaxation of MWBIS, is going to give us a better factor approximation algorithm for the MWBIS problem on bipartite graphs. 

\subsection{Notational note}
Throughout the paper, we consider only finite, undirected, and simple graphs. For a graphs $G$, we shall use $V(G)$ to denote its vertex set and $E(G)$ to denote its edge set. For any $S \subseteq V(G)$, we shall use $|S|$ to denote the cardinality of the set $S$. If a weight function $w:V(G) \rightarrow \mathbb{R^+}$ is given, then $||S||$ shall be used to denote the sum of the weights of the vertices in $S$. Otherwise, $||S||$ will denote the number of edges in $G$ having at least one endpoint in $S$. For any vertex $v$ in the graph under consideration, we shall use $deg(v)$ to denote its degree in the graph. For any positive integer $n$, we shall use $[n]$ to denote the set $\{1, \ldots , n\}$.  
\section{MWBIS in bipartite graphs}
\label{sec:bip}
We begin by first showing that D-MIVC 
on bipartite graphs is NP-hard in Section \ref{subsec:NPhard}. We give a $1/2$ factor approximation algorithm for MWBIS in Section \ref{subsec:halfApprox} and show in Section \ref{subsec:bipartiteIntGap} that LP-based methods are unlikely to yield better approximation algorithms. 
\subsection{NP-hardness of D-MIVC in bipartite graphs}
\label{subsec:NPhard}
We show that D-MIVC problem on bipartite graphs is NP-hard by reducing from $(n-4,k)$-CLIQUE problem. The $(n-4,k)$-CLIQUE problem is the famous clique problem on a graph, where the graph under consideration is an $(n-4)$-regular graph on $n$ vertices. \\
$(n-4,k)$-CLIQUE = $\{ <G,k>~|~G$ is an $(n-4)$-regular graph on  $n$  vertices that contains a clique of size at least  $k$,  where  $k < \frac{n}{2}\}$. \\
The $(3,k)$-Independent Set ($(3,k)$-IS) problem is about deciding whether there exists an independent set of size at least $k$ or not in a given $3$-regular graph. It was shown in \cite{garey1976some} that the $(3,k)$-IS problem is NP-hard. This means that the $(n-4,k)$-CLIQUE problem is also NP-hard. Note that when $k \geq \frac{n}{2}$, the $(3,k)$-IS problem is about deciding whether the given graph is bipartite or not and is therefore polynomial-time solvable. Below we outline our reduction from the $(n-4,k)$-CLIQUE problem to D-MIVC.

\begin{construction} 
\label{construction:bipartite}
Let $n$ be a positive integer greater than $11$ and let $r=n-4$. Given an $r$-regular graph $G$ with $V(G) = \{v_1, \ldots , v_n\}$ and $E(G) = \{e_1, \ldots , e_m\}$, we construct a bipartite graph $H$ with bipartition $\{A,B\}$, where $A = \{a_1, \ldots , a_m\}$, $B = \beta \cup \Pi$ with $\beta = \{b_1, \ldots , b_n\}$ and $\Pi = \{p_{i,j,}~|~i\in [m], j \in [r-3]\}$. The edge set  $E(H) = \{ a_ip_{i,j}~|~i \in [m], j \in [r-3]\} \cup \{a_ib_j~|~\mbox{edge } e_i \mbox{ is incident on vertex } v_j \mbox{ in } G\}$. Note that, $deg(a_i) = 2 + (r-3) = r-1$, $deg(b_i) = r$, and every $p_{i,j}$ is a pendant vertex.  
\end{construction}
\begin{observation}
\label{obv:incidence}
In Construction \ref{construction:bipartite}, the subgraph of $H$ induced on the vertex set $A \cup \beta$ is isomorphic to the incidence graph of $G$. Thus, $H$ is isomorphic to the incidence graph of $G$ with $r-3$ pendant vertices (, namely the $p_{i,j}$ vertices) hanging down from each $a_i$. 
\end{observation}

In Lemma \ref{lem:NP-hard-Bip-1}, we prove that if  the graph $G$ contains a clique of size $k$, then there exist some $k+x$ independent vertices in $H$ that cover at least $kr + x(r-1)$ edges, where $x=m - (kr - {k \choose 2})$. In Lemma \ref{lem:NP-hard-Bip-2}, we prove the reverse implication of this statement. 
\begin{lemma}
\label{lem:NP-hard-Bip-1}
If there exists a clique of size $k$ in the $r$-regular graph $G$ given in Construction \ref{construction:bipartite}, then there exists $k+x$ independent vertices in the bipartite graph $H$ constructed that cover at least $kr + x(r-1)$ edges, where $x=m - (kr - {k \choose 2})$. 
\end{lemma}
\begin{proof}
Without loss of generality, let us assume that the set of vertices $S= \{v_1, \ldots , \\ v_k\}$ form a $k$-clique in $G$. It is easy to verify that $||S|| = kr - {k \choose 2}$ and therefore the number of edges not incident on the vertices of $S$ in $G$ is $x = m - ||S|| = m - (kr - {k \choose 2})$. Without loss of generality, let these $x$ edges be $e_1, \ldots , e_x$. Let $T = \{a_1, \ldots , a_x, b_1, \ldots , b_k\}$. Clearly, $T$ is a set of independent vertices in $H$ and $||T|| = kr + x(r-1)$.  
\end{proof}
Before we prove Lemma \ref{lem:NP-hard-Bip-2}, we prove a few supporting lemmas, namely Lemmas \ref{lemma:support1}
to \ref{lemma:support3}, that give us some insight into the structure of $G$ and $H$.   
\begin{lemma}
\label{lemma:support1}
Suppose there exists no clique of size $k$ in the $r$-regular graph $G$ given in Construction \ref{construction:bipartite}. Let  $i$ be an integer such that $0 \leq i \leq r-k$. Let $p_{k+i}(G) = \max\{m-||S||~:~S \subseteq V(G), |S| = k+i\}$. Then, \\
(i) $p_k(G) \leq x-1$, \\
(ii) for every $i \in [r-k]$, $p_{k+i}(G) \leq p_{k+i-1}(G) - (r - (k+i-2))$. 
\end{lemma}
\begin{proof}
For every $i$, let $S_{k+i}$ (, which is a subset of $V(G)$,) be a set such that $|S_{k+i}| = k+i$ and $m - ||S_{k+i}|| = p_{k+i}(G)$. \\
(i) Since $G$ does not contain any $k$-clique, $||S_k|| \geq kr - ({k \choose 2} - 1)$ and thus $p_k(G) = m - ||S_k|| \leq m - (kr - ({k \choose 2} - 1)) = x - 1$. \\
(ii) Let $i \in [r-k]$. Since $G$ does not contain any $k$-clique, $\exists v \in S_{k+i}$ such that $v$ has at most $k+i-2$  neighbours in $S_{k+i}$ and therefore at least $r- (k+i-2)$ neighbours outside $S_{k+i}$. Let $X_{k+i-1} = S_{k+i} \setminus \{v\}$. Then, $||S_{k+i}|| \geq ||X_{k+i-1}|| + r - (k+i-2) \geq ||S_{k+i-1}|| + r - (k+i-2)$ (from the definition of $S_{k+i-1}$). Therefore, $p_{k+i}(G) = m - ||S_{k+i}|| \leq m - (||S_{k+i-1}|| + r - (k+i-2)) \leq p_{k+i-1}(G) - (r - (k+i-2))$. 
\end{proof}	
Below, we  introduce a couple of new definitions. 
Consider the bipartite graph $H$ with bipartition $\{A,B\}$ (, where $B = \beta \uplus \Pi$,) constructed from $G$ in Construction \ref{construction:bipartite}. Let $i$ be an integer such that $0 \leq i \leq r-k$. 
\begin{definition}
\label{defn:Q}
For any set $S \subseteq \beta$, let $q(S) := \max \{|X|~:~X \subseteq A \mbox{ and no vertex in } X \\ \mbox{is adjacent with any vertex in } S\}$. Then, we define $q_{k+i}(H) := \max \{q(S)~:~S \mbox{ is a }\\ (k+i) \mbox{-sized subset of } \beta\}$. 
\end{definition}
\begin{definition}
\label{def:S}
Let $\mathcal{I}_{k+i} = \{I \subseteq V(H)~|~I \mbox{ is a } (k+x) \mbox{-sized independent set, and} |I \cap \beta| = k+i\}$, where $x=m - (kr - {k \choose 2})$. Then, we define $s_{k+i}(H) := \max \{||I||~:~I \in \mathcal{I}_{k+i}\}$. 
\end{definition}

Suppose the graph $G$ does not contain any $k$-clique. Then, from Observation \ref{obv:incidence}, Lemma \ref{lemma:support1}, and Definition \ref{defn:Q}, we have $p_{k+i}(G) = q_{k+i}(H)$. This gives us the following lemma. 
\begin{lemma}
\label{lemma:support2}
Suppose there exists no clique of size $k$ in the $r$-regular graph $G$ given in Construction \ref{construction:bipartite}. Consider the bipartite graph $H$ constructed from $G$. 
We have, 
\\
(i) $q_k(H) \leq x-1$, \\
(ii) for every $i \in [r-k]$, $q_{k+i}(H) \leq q_{k+i-1}(H) - (r-(k+i-2))$.  
\end{lemma}

\begin{lemma}
\label{lemma:support3}
Suppose there exists no clique of size $k$ in the $r$-regular graph $G$ given in Construction \ref{construction:bipartite}. Consider the bipartite graph $H$ constructed from $G$. Then, $s_r(H) < s_{r-1}(H) < \cdots <  s_k(H) < kr + x(r-1)$, where $x=m - (kr - {k \choose 2})$. 
\end{lemma}
\begin{proof}
Let $i$ be an integer such that $0 \leq i \leq r-k$. Let $S_{k+i}$ be a $(k+x)$-sized independent set in $H$ having exactly $k+i$ vertices from the set $\beta$ such that $||S_{k+i}||$ is maximum among all such independent sets. That is, $||S_{k+i}|| = s_{k+i}(H)$ (from the definition of $s_{k+i}(H)$). 
Recall that $S_{k+i} \subseteq V(H) = A \uplus \beta \uplus \Pi$ and from its definition we have $|S_{k+i} \cap (A \cup \Pi)| = x-i$. Since the degree of any vertex in $A$ (which is $r-1$) is greater than the degree of any vertex in $\Pi$ (which is $1$), we have $||S_{k+i}||$ maximized when $S_{k+i} \cap A$ is maximized. 
Then, from the definition of $S_{k+i}$ and by Definition \ref{defn:Q}, we get $|S_{k+i} \cap A| = q_{k+i}(H)$. Therefore, 
\begin{eqnarray}
s_{k+i}(H) & = & ||S_{k+i}|| \nonumber \\
& = & ||S_{k+i} \cap \beta|| + ||S_{k+i} \cap A|| + ||S_{k+i} \cap \Pi|| \nonumber \\
& = & |S_{k+i} \cap \beta|r + |S_{k+i} \cap A|(r-1) + |S_{k+i} \cap \Pi| \nonumber \\
& = & (k+i)r + (q_{k+i}(H))(r-1) + (k+x) - \left(k+i + \left(q_{k+i}(H)\right)\right). \label{eqn:s_k}
\end{eqnarray}
By Statement (i) of Lemma \ref{lemma:support2}, we have $q_k(H) = x-a$, where $a \geq 1$. Therefore, $s_k(H) = kr + (x-a)(r-1) + a \leq kr + (x-1)(r-1) + 1 < kr + x(r-1)$.

Let $i$ be an integer such that $0 \leq i < r-k$. Then,  
\begin{eqnarray*}
s_{k+i}(H) - s_{k+i+1}(H) & = & (q_{k+i}(H) - q_{k+i+1}(H))(r-2) - (r-1) \mbox{\hspace{0.1in} (from Eqn. (\ref{eqn:s_k}))}\\
 &  \geq & (r-(k+i-1))(r-2) - (r-1)  \mbox{\hspace{0.1in} (by (ii) in Lemma \ref{lemma:support2})} \\
 & \geq & 2(r-2) - (r-1) \mbox{\hspace{0.1in} (since } i \leq r -k - 1 \mbox{)} \\
 & \geq & 5 \mbox{ \hspace{0.1in} (since } r = n-4 \mbox{ and } n > 11 \mbox{).} 
\end{eqnarray*}
\end{proof}

\begin{lemma}
\label{lem:NP-hard-Bip-2}
If the $r$-regular graph $G$  given in Construction \ref{construction:bipartite} does not contain any clique of size $k$, then no set of $k+x$ independent vertices in the bipartite graph $H$ constructed covers $kr + x(r-1)$ edges or more, where $x=m - (kr - {k \choose 2})$.
\end{lemma}
\begin{proof}
Suppose the graph $G$  given in Construction \ref{construction:bipartite} does not contain any clique of size $k$. We shall then show that no set of $k+x$ independent vertices in the bipartite graph $H$ constructed from $G$ covers $kr + x(r-1)$ edges or more

Let $Y$ be an independent set of vertices of cardinality at most $k+x$ in $H$. Then, $Y = Y_A \uplus Y_{\beta} \uplus Y_{\Pi}$, where $Y_A = Y \cap A$, $Y_{\beta} = Y \cap \beta$, and $Y_{\Pi} = Y \cap \Pi$. The contribution of any vertex $v \in Y$ to the coverage $||Y||$ is its degree and this is highest when $v \in Y_{\beta}$ (then, $deg(v) = r$) and lowest when $v \in Y_{\Pi}$ (in this case, $deg(v) = 1$). When $v \in Y_A$, $deg(v) = r-1$. Below, we shall prove that $||Y|| < kr + x(r-1)$. The proof is split into $2$ cases based on the cardinality of $Y_{\beta}$. 
\begin{case}[$|Y_{\beta}| < k$]
\label{case:easy} 
Let $|Y_{\beta}| = k - a$, where $a \geq 1$. Then, $||Y|| = |Y_{\beta}|r + |Y_A|(r-1) + |Y_{\Pi}|  \leq (k-a)r + (x+a)(r-1) \leq (k-1)r + (x+1)(r-1) = kr + x(r-1) - 1$. 
\end{case}
\begin{case}[$|Y_{\beta}| \geq k$]
We split this case into two subcases:- \\
(i) $|Y_{\beta}| \leq r$ ($= n-4$). \\
Let $0 \leq i \leq  r-k$. Let $|Y_{\beta}| = k+i$. Then, by Definition \ref{def:S} and  Lemma \ref{lemma:support3}, $||Y|| \leq s_{k+i} < kr + x(r-1)$. 
\\(ii) $|Y_{\beta}| > r ~(= n- 4)$. \\
Since $|\beta| = n$, by definition of $Y_{\beta}$, $|Y_{\beta}| \leq n$. Let $i \in \{0, \ldots , 3\}$. Let $|Y_{\beta}| = n-i$. 
The set $Y$ being an independent set in $H$, by Observation \ref{obv:incidence}, the vertices in $Y_A$ correspond to edges that are `outside' the set of vertices in $G$ that correspond to $Y_{\beta}$. As the cardinality of $Y_{\beta}$ is close to $n$ (in fact, at least $n-3$), we get $|Y_A| \leq i$.   
We are now prepared to estimate the coverage of $Y$. 
\begin{eqnarray*}
||Y|| & = & |Y_{\beta}|r + |Y_A|(r-1) + |Y_{\Pi}| \\ 
& \leq & (n-i)r + i(r-1) + \left((k+x) - n\right) \\
 & \leq & nr + k+x - n \hspace{0.2in} \mbox{(since }i=0\mbox{ yields the highest value)}
\end{eqnarray*} 
In order to prove subcase (ii), we need to show that 
\begin{eqnarray*}
kr + x(r-1) & > & nr + k+x - n \\
\mbox{i.e. }  x(r-2) &  >  &(n-k)(r-1) \mbox{\hspace{0.2in} (rearranging terms)} \\
\mbox{i.e. } x & >  &(n-k)\left(1  + \frac{1}{r-2}\right) 
\end{eqnarray*} 
Since $r=n-4$ and $n>11$, we have $r \geq 8$. Thus, in order to prove this subcase it is enough to prove the following inequality:-
\begin{eqnarray*}
\label{ineq:toProve}
x  &> & \frac{7}{6}(n-k).  
\end{eqnarray*}
We know that $x  =  m - \left(kr - {k \choose 2}\right) = \frac{nr}{2} - \left(kr - \frac{k^2-k}{2}\right) = \frac{nr - 2kr + k^2 - k}{2} = \frac{r(n-k) - kr + k^2 - k}{2} \\ = \frac{r(n-k) - k(n-4) + k^2 - k}{2} = \frac{(r-k)(n-k) + 3k}{2} > \frac{((n-4)-k)(n-k)}{2} \geq \frac{(n-7)(n-k)}{4}$ (since $k \leq \frac{n-1}{2}$. See the definition of $(n-4,k)$-CLIQUE in the beginning of this section.). Since $n>11$, we get $x \geq \frac{5}{4}(n-k) > \frac{7}{6}(n-k)$.  

\end{case}

\end{proof}   

Lemma \ref{lem:NP-hard-Bip-1} and Lemma \ref{lem:NP-hard-Bip-2} imply the following theorem. 
\begin{theorem}
\label{thm:bipNPhard}
D-MIVC problem on bipartite graphs is NP-hard. 
\end{theorem}   
\subsection{A $\frac{1}{2}$-approximation algorithm}
\label{subsec:halfApprox}
Consider the following greedy approximation algorithm for the MWBIS problem in bipartite graphs. 
\begin{algorithm}[H]
\label{algo:halfapprox}
\caption{$\frac{1}{2}$-approximation algorithm 
for the MWBIS problem on bipartite graphs}
\begin{algorithmic}[1]
\INPUT A bipartite graph $G$ with bipartition $\{A,B\}$, a weight function $w:V(G) \rightarrow \mathbb{R}^+$, and a positive integer $k$.
\OUTPUT An independent set of size at most $k$. 

\item Let $S_{A}$ denote the set of $k$ highest weight vertices in $A$. In the case when $k > \vert A \vert$, $S_{A} = A$. Similarly, let $S_{B}$ denote the set of $k$ highest weight vertices in $B$. In the case when $k > \vert B \vert$, $S_{B} = B$. Find $S_A$ and $S_B$. 

\item If $\parallel S_{A} \Vert \geq \Vert S_{B} \Vert $, return $S_{A}$. Otherwise, return $S_{B}$. 

\end{algorithmic}
\end{algorithm}
We claim that Algorithm $1$ is a $\frac{1}{2}$-factor approximation algorithm 
for MWBIS in bipartite graphs. Without loss of generality, let us assume that Algorithm $1$ returns the set $S_{A}$. Let $O$ be an optimal solution for the MWBIS problem on the bipartite graph $G$ under consideration. Let $O_{A}$ = $O \cap A$ and $O_{B} = O \cap B$. Then $\Vert O \Vert$ = $\Vert O_{A} \Vert $ + $\Vert O_{B} \Vert$ $\leq $ $\Vert S_{A} \Vert $ + $\Vert S_{B} \Vert$ $\leq$ 2$\Vert S_{A} \Vert$. Hence, we prove the claim. It is easy to see that Algorithm $1$ runs in $O(nk)$ time. 

\begin{figure}[!ht]
\begin{center}
\begin{pspicture}(-0.5,-0.5)(3,8.5)
 	\psline[showpoints=true](0,1)(2,0)
	\psline[showpoints=true](0,1)(2,0.5)
	\psline[linestyle=dotted](1,1)(1,1.3)
	\psline[showpoints=true](0,1)(2,2)	
	
 	\psline[showpoints=true](0,3)(2,2)
	\psline[showpoints=true](0,3)(2,2.5)
	\psline[linestyle=dotted](1,3)(1,3.3)
	\psline[showpoints=true](0,3)(2,4)
	
	\psline[linestyle=dotted](0,5)(0,5.3)		
	
	\psline[showpoints=true](0,7)(2,6)
	\psline[showpoints=true](0,7)(2,6.5)
	\psline[linestyle=dotted](1,7)(1,7.3)
	\psline[showpoints=true](0,7)(2,8)
	
	\uput[l](0,1){$a_1$}
	\uput[l](0,3){$a_2$}
	\uput[l](0,7){$a_{\frac{k}{2}}$}
	
	\uput[r](2,0){$b_{1,1}$}
	\uput[r](2,0.5){$b_{1,2}$}
	\uput[r](2,2){$b_{1,x} (= b_{2,1})$}
	\uput[r](2,2.5){$b_{2,2}$}
	\uput[r](2,4){$b_{2,x} (= b_{3,1})$}
	\uput[r](2,6){$b_{\frac{k}{2}-1,x} (= b_{\frac{k}{2},1})$}
	\uput[r](2,6.5){$b_{\frac{k}{2},2}$}
	\uput[r](2,8){$b_{\frac{k}{2},x}$}
	
\end{pspicture}
\end{center}
\caption{Graph $H_1$}
\label{fig:H1}
\end{figure}
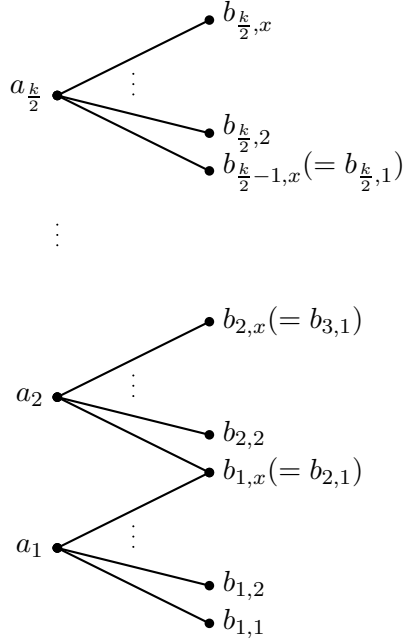

\subsubsection{Tight example}
Below we describe the construction of a bipartite graph $H$ that is a tight example to our analysis of Algorithm $1$. Let $H_{1}$ be the connected, bipartite graph given in Figure \ref{fig:H1}. Let $H_2$ be a graph isomorphic to $H_1$ (one can imagine $H_2$ to be a mirror copy of Figure \ref{fig:H1} drawn right above it). We obtain the connected, bipartite graph $H$ from the union of $H_1$ and $H_2$ by adding an edge connecting a pendant vertex in $H_1$ with a pendant vertex in $H_2$. More formally, $V(H) = V(H_1) \cup V(H_2)$ and $E(H) = E(H_1) \cup E(H_2) \cup \{h_1,h_2\}$, where $h_1$ and $h_2$ are any two pendant vertices in $H_1$ and $H_2$, respectively. 
We claim that $H$ is a tight example for Algorithm $1$. Consider the MIVC problem on graph $H$ with parameter $k$. The optimal solution chooses all the degree $x$ vertices. This gives a solution of size $kx$. Since each part in $H$ has exactly $\frac{k}{2}$ vertices of degree $x$ and the remaining vertices are each of degree at most $2$, our algorithm returns a solution of size at most $\frac{kx}{2} + \frac{k}{2}(2)$ = $\frac{kx}{2} + k$. Thus, the approximation ratio is at most $\frac{\frac{kx}{2} +k}{kx}$ = $\frac{1}{2} + \frac{1}{x}$. This ratio approaches $\frac{1}{2}$ as $x$ tends to infinity. 

\subsubsection{Extending Algorithm $1$ to general graphs}
\label{subsec:ExtBipAlgo}
Algorithm $1$ can be generalized to an arbitrary graph $G$ (that is not necessarily bipartite) provided we are given a proper vertex coloring of $G$. Suppose we are given a proper coloring of the vertices of $G$ using $p$ colors that partitions the vertices into $p$ color classes namely $C_{1},\ldots , C_{p}$. Let $S_{1},\ldots,S_{p}$ denote sets of $k$ highest degree vertices in color classes $C_{1}, \ldots , C_{p}$, respectively. For any $S_{i}$, if $k > \vert C_{i}\vert$, then $S_{i} = C_{i}$. Find a set $S_{i}$ such that $\Vert S_{i} \Vert \geq \Vert S_{j} \Vert$, $\forall j \in [p]$. Return $S_{i}$. 

This algorithm gives us a $\frac{1}{p}$-approximation 
for MWBIS. Suppose it returns a set $S_{j}$ as the output. Let $O$ be an optimal solution for the MWBIS problem on $G$. Let $O_{i}$ = $O \cap C_{i}$ for $i \in [p]$. Then $\Vert O \Vert$ = $ \Vert O_{1} \Vert+ \ldots + \vert O_{p} \Vert $ $\leq$ $p \Vert S_{j}\Vert$. It is easy to see that the algorithm runs in $O(nk)$ time. 

We know that we can properly color any graph $G$ with $\Delta +1$ colors using a greedy coloring algorithm, where $\Delta$ denotes the maximum degree of $G$. This gives us a $\frac{1}{\Delta +1}$-approximation algorithm for MWBIS on $G$. Since $d$-degenerate graphs can be properly colored using $d+1$ colors in linear time,  this algorithm  is a $\frac{1}{d+1}$-factor approximation algorithm for such graphs.

\subsection{LP relaxation and a matching integrality gap}
\label{subsec:bipartiteIntGap}
Let us take another look at Integer Program (\ref{IP:MIVC}) given in Section \ref{subsec:IPforMWBIS} for the MWBIS problem on $G$ (, where $G$ is a graph with $V(G) = \{v_1, \ldots , v_n\}$ and $w:V(G) \rightarrow \mathbb{R}^+$). Recall, we use $\mathcal{C}$ to denote the set of all maximal cliques in $G$. We obtain the LP relaxation of this program by changing the domain of variable $x_i$ from $x_i \in \{0,1\}$ to $x_i \geq 0$. 
\begin{align}
Maximize &\sum_{i \in [n]} w(v_i)\cdot x_{i} \label{LP:MIVC} \\
&s.t. \sum_{i=1}^{n} x_{i} \leq k  \nonumber \\
&\sum_{i:v_i \in C}x_i \leq 1, \quad \forall C \in \mathcal{C} \nonumber \\
&x_{i} \geq 0 , \quad \forall i \in [n]. \nonumber 
\end{align}
We now present an instance of the MIVC problem to illustrate that the integrality gap of the above LP is upper bounded by $\frac{1}{2} + \epsilon$ (where $\epsilon$ is any number greater than $0$) for bipartite graphs. 
\begin{example}
\label{example:BipIntGap}
Let $H$ be a bipartite graph with bipartition $\lbrace A, B \rbrace$, where $\vert B \vert$ = $k(k-1) +1 = p$ and $\vert A \vert$ = $(k -1)p + 1$. Let $A =  \lbrace a_{i,j}~|~i \in [p], j \in[k-1] \rbrace \cup \{a_{0}\}$ and $B = \lbrace b_{1}, \ldots , b_{p}\rbrace$. The edge set $E = \lbrace a_{0}b_{i}~|~ i \in [p] \rbrace$ $\cup$ $\lbrace a_{i,j}b_{i}~|~ i \in [p], j \in [k-1]\rbrace$. Note that $\deg(a_{0}) = p$, $\deg(b_{i}) = k$,  and $\deg(a_{i,j}) = 1$. For every $v \in V(H)$, let $w(v) = deg(v)$. Figure \ref{fig:BipIntGap} illustrates a drawing of $H$ with $k=3$. 
 \end{example}
The best integral solution to the above LP for $H$ is to select vertex $a_{0}$ and $k-1$ vertices  of degree $1$ from the set $\lbrace a_{i,j}~|~i \in [p], j \in[k-1] \rbrace$. The coverage by these $k$ vertices is $deg(a_{0}) + (k-1) = p + (k-1) = k(k-1) +1 + (k-1) =  k^{2}$.  

The $x$-variables in the LP have the following associations: $x_0$ is associated with vertex $a_0$, $x_i$ with $b_i$, and $x_{i,j}$ with $a_{i,j}$. One possible fractional solution to the above LP for $H$ is to assign $x_{0}$ = $\frac{k-1}{k}$, $x_{i} = \frac{1}{k}$, $\forall i \in [p]$, and $x_{i,j} = 0, \forall i \in [p], \forall j \in [k-1]$. 
This is a feasible solution to the above LP as $x_0 + \sum_{i \in [p], j \in [k-1]}x_{i,j}+ \sum_{i \in p}x_{i} = \frac{k-1}{k} + 0 + p\frac{1}{k} =k$ and sum of the $x$-values of every maximal clique is at most $1$. 
The number of edges covered by this fractional solution is $(\frac{k-1}{k} + \frac{1}{k})p + (\frac{1}{k}  p (k-1)) = 
p (1 + \frac{k-1}{k}) = (k(k-1) + 1 )(\frac{2k-1}{k})$. Hence the upper bound for integrality gap of the above LP established by this example is $\frac{k^{2}}{(k(k-1) + 1 )(\frac{2k-1}{k})} = \frac{k^3}{(k^{2}-k+1 )(2k-1)} = \frac{k^{3}}{2k^{3} - 3k^{2} + 3k-1} = \frac{1}{2- \frac{3}{k} + \frac{3}{k^{2}} - \frac{1}{k^{3}}} = \frac{1}{2- \frac{1}{k}(3 + \frac{1}{k^2} - \frac{3}{k}) }$. This value approaches $1/2$ as $k$ tends to infinity. Thus, the integrality gap of the above LP is upper bounded by $\frac{1}{2} + \epsilon$ (where $\epsilon$ is any number greater than $0$) for bipartite graphs. \\

\begin{figure}[!ht]
\begin{center}
\begin{pspicture}(0,-0.5)(14,2.5)
 	\psline[showpoints=true](1,2)(0,0)
	\psline[showpoints=true](1,2)(1,0)
	\psline[showpoints=true](1,2)(2,0)	
	
 	\psline[showpoints=true](3,2)(0,0)
	\psline[showpoints=true](3,2)(3,0)
	\psline[showpoints=true](3,2)(4,0)	
	
 	\psline[showpoints=true](5,2)(0,0)
	\psline[showpoints=true](5,2)(5,0)
	\psline[showpoints=true](5,2)(6,0)	
	
 	\psline[showpoints=true](7,2)(0,0)
	\psline[showpoints=true](7,2)(7,0)
	\psline[showpoints=true](7,2)(8,0)	
	
 	\psline[showpoints=true](9,2)(0,0)
	\psline[showpoints=true](9,2)(9,0)
	\psline[showpoints=true](9,2)(10,0)	
	
 	\psline[showpoints=true](11,2)(0,0)
	\psline[showpoints=true](11,2)(11,0)
	\psline[showpoints=true](11,2)(12,0)	
	
 	\psline[showpoints=true](13,2)(0,0)
	\psline[showpoints=true](13,2)(13,0)
	\psline[showpoints=true](13,2)(14,0)							
	
	\uput[u](1,2){$b_1$}
	\uput[u](3,2){$b_2$}
	\uput[u](5,2){$b_3$}
	\uput[u](7,2){$b_4$}
	\uput[u](9,2){$b_5$}
	\uput[u](11,2){$b_6$}
	\uput[u](13,2){$b_7$}
	
	\uput[270](0,0){$a_0$}

	\uput[270](1,0){$a_{1,1}$}
	\uput[270](2,0){$a_{1,2}$}
	
	\uput[270](3,0){$a_{2,1}$}
	\uput[270](4,0){$a_{2,2}$}
	
	\uput[270](5,0){$a_{3,1}$}
	\uput[270](6,0){$a_{3,2}$}
	
	\uput[270](7,0){$a_{4,1}$}
	\uput[270](8,0){$a_{4,2}$}
	
	\uput[270](9,0){$a_{5,1}$}
	\uput[270](10,0){$a_{5,2}$}
	
	\uput[270](11,0){$a_{6,1}$}
	\uput[270](12,0){$a_{6,2}$}
	
	\uput[270](13,0){$a_{7,1}$}
	\uput[270](14,0){$a_{7,2}$}	
		
\end{pspicture}
\end{center}
\caption{A drawing of the bipartite graph $H$ described in Example \ref{example:BipIntGap} with $k=3$.} 
\label{fig:BipIntGap}
\end{figure}
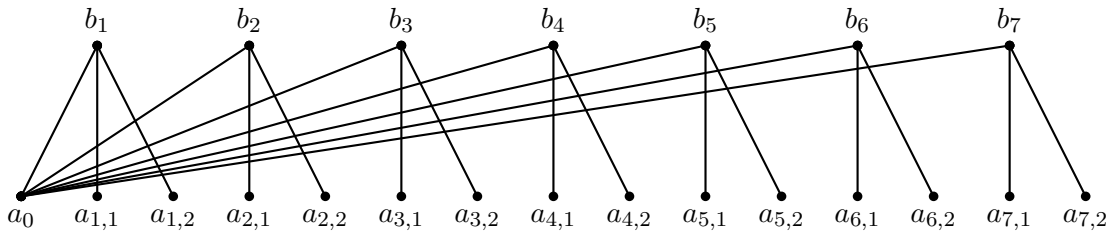

\section{Concluding remarks}
In this paper we study the MWBIS problem, a generalization of MWIS, mainly in the context of bipartite graphs. In Theorem \ref{thm:bipNPhard}, we prove that MWBIS on bipartite graphs is NP-hard and then we propose an easy, greedy $\frac{1}{2}$-approximation algorithm, whose factor of approximation matches with the integrality gap of a natural LP relaxation of the problem. 
It would be interesting to see if one can improve the factor of approximation using a more clever and sophisticated approach.

It is well-known that the MWIS problem in general graphs does not have a polynomial time constant factor approximation algorithm, unless P=NP. H{\aa}stad \cite{Has99} showed that there is no $\frac{1}{n^{1-\epsilon}}$-factor approximation algorithm for MWIS, where $\epsilon > 0$, assuming NP-hard problems have no randomized polynomial time algorithms. For every sufficiently large $\Delta$, there is no $\Omega(\frac{\log^2\Delta}{\Delta})$-factor polynomial time approximation algorithm for MWIS in a degree-$\Delta$ bounded graph, assuming the unique games conjecture and $P \neq NP$ \cite{austrin2009inapproximability}. As MWBIS is a generalization of MWIS, these results hold  true for the MWBIS problem too. Let $G$ be a graph on $n$ vertices with the degree of  any vertex being at most $\Delta$. Now, suppose we have a $t$-factor approximation algorithm $A(G)$ for finding MWIS in $G$. Then, we can obtain a $\frac{kt}{n}$-factor approximation algorithm for MWBIS in $G$. Let $A(G)$ return an independent set $S$ as output. Let $S'$ denote the set of $k$ highest weight vertices in $S$. If $k > |S|$, then $S' = S$. Return $S'$. Let $O$, $O'$ be optimal solutions for the MWIS and the MWBIS problem in $G$, respectively. We have $||S'|| \geq \frac{k||S||}{|S|} \geq \frac{k||S||}{n} \geq \frac{kt||O||}{n} \geq \frac{kt||O'||}{n}$. Boppana and Halld{\'o}rsson in \cite{boppana1992approximating} gave an $\Omega(\log^2n/n)$-factor approximation algorithm for MWIS.  This yields an $\Omega(k\log^2n/n^2)$-factor algorithm for MWBIS. For bounded degree graphs, Halld{\'o}rsson in \cite{halldorsson2000approximations} gave an $\Omega(\sqrt{\log \Delta}/\Delta)$-factor approximation algorithm for MWIS. This gives  an $\Omega(k\sqrt{\log \Delta}/n\Delta)$-factor approximation algorithm for MWBIS. Finding better approximation algorithms for MWBIS in general graphs is another nice open question to pursue. 

\bibliographystyle{plain}

\end{document}